\documentclass[10pt]{article}

\usepackage{amssymb}
\usepackage{amsfonts}
\usepackage{amsmath}
\usepackage{amscd}
\usepackage[all,cmtip]{xy}
\usepackage{amsthm}
\usepackage{setspace}
\usepackage{enumerate}
\usepackage{geometry}
\usepackage{tcolorbox}
\usepackage{hyperref}
\usepackage{tikz}
\usepackage{pgfplots}

\usetikzlibrary{shapes.geometric}
\usetikzlibrary{math}

\usepackage{pgfplots,pgfplotstable}
\usetikzlibrary{positioning}

\usepackage{booktabs}

\usepackage{array}
\newcommand{\PreserveBackslash}[1]{\let\temp=\\#1\let\\=\temp}
\newcolumntype{C}[1]{>{\PreserveBackslash\centering}p{#1}}

\hypersetup{
allbordercolors=blue, pdfborder=0 0 1.3}

\theoremstyle{plain}
\newtheorem{thm}{Theorem}

\theoremstyle{definition}
\newtheorem{rem}{Remark}
\newtheorem{defi}{Definition}

\usepackage[bottom]{footmisc}

\title{The Toda-Weyl mass spectrum}
\author{Martin T.  Luu}

\date{}

\pgfplotsset{compat=1.17} 

\begin{document}

\maketitle

\begin{abstract} 
The masses of affine Toda theories are known to correspond to the entries of a Perron-Frobenius eigenvector of the relevant Cartan matrix. The Lagrangian of the theory can be expressed in terms of a suitable eigenvector of a Coxeter element in the Weyl group. We generalize this set-up by formulating Lagrangians based on eigenvectors of arbitrary elements in the Weyl group. Under some technical conditions (that hold for many Weyl group elements), we calculate the classical mass spectrum. In particular, we indicate the relation to the relative geometry of special roots, generalizing the affine Toda mass spectrum description in terms of the Cartan matrix.  Related questions of three point coupling and integrability are left to be addressed on a future occasion.
\end{abstract}

\section{Introduction}
\label{introduction-section}

In low temperatures,  below around $2.95$ K,  Cobalt-Niobate $\textrm{Co} \textrm{Nb}_{2} \textrm{O}_6$ is in a magnetically ordered phase,  see \cite{SCH}.  The magnetic structure is dominated by the Cobalt ions $\textrm{Co}^{2+}$ since these possess the only unpaired electrons, having $7$ electrons in the $3d$ orbital. The solid exhibits quasi-$1$D behavior due to the chain geometry in which the Cobalt ions assemble.  Describing the $\textrm{Co} \textrm{O}_{6}$ environment, the oxygen atoms assemble in a slightly deformed hexagonal octahedron around zig-zag chains of Cobalt ions. Schematically:
  
$$
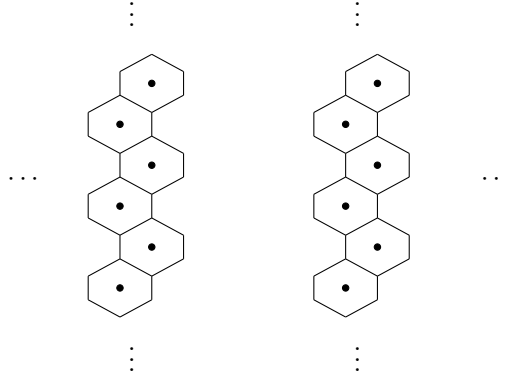
\begin{tikzpicture}
\tikzmath{
\X=0;
\Y=0;
\x=3;
\y=0;
\J=0.58;
\j=(0.48)*\J;
\xhori=2*sin(68.85/2)*\j;
\yhori=2*sin(68.85/2)*\j;
\ydiag=cos(68.85/2)*\j;
\xdiag=\J-sin(68.85/2)*\j;
coordinate \P;
coordinate \Q;
coordinate \R;
coordinate \S;
coordinate \T;
coordinate \U;
coordinate \p;
coordinate \q;
coordinate \r;
coordinate \s;
coordinate \t;
coordinate \u;
\P1=(\X,\Y);
\P2=(\X-1*\xdiag,\Y+\ydiag);
\P3=(\X-1*\xdiag,\Y+\ydiag+\yhori);
\P4=(\X-0*\xdiag,\Y+2*\ydiag+\yhori);
\P5=(\X+1*\xdiag,\Y+1*\ydiag+\yhori);
\P6=(\X+1*\xdiag,\Y+1*\ydiag);
\Q1=(\X-0*\xdiag,\Y+2*\ydiag+2*\yhori);
\Q2=(\X+1*\xdiag,\Y+3*\ydiag+2*\yhori);
\Q3=(\X+2*\xdiag,\Y+2*\ydiag+2*\yhori);
\Q4=(\X+2*\xdiag,\Y+2*\ydiag+1*\yhori);
\R1=(\X-1*\xdiag,\Y+3*\ydiag+2*\yhori);
\R2=(\X-1*\xdiag,\Y+3*\ydiag+3*\yhori);
\R3=(\X-0*\xdiag,\Y+4*\ydiag+3*\yhori);
\R4=(\X+1*\xdiag,\Y+3*\ydiag+3*\yhori);
\S1=(\X-0*\xdiag,\Y+4*\ydiag+4*\yhori);
\S2=(\X+1*\xdiag,\Y+5*\ydiag+4*\yhori);
\S3=(\X+2*\xdiag,\Y+4*\ydiag+4*\yhori);
\S4=(\X+2*\xdiag,\Y+4*\ydiag+3*\yhori);
\T1=(\X-1*\xdiag,\Y+5*\ydiag+4*\yhori);
\T2=(\X-1*\xdiag,\Y+5*\ydiag+5*\yhori);
\T3=(\X-0*\xdiag,\Y+6*\ydiag+5*\yhori);
\T4=(\X+1*\xdiag,\Y+5*\ydiag+5*\yhori);
\U1=(\X-0*\xdiag,\Y+6*\ydiag+6*\yhori);
\U2=(\X+1*\xdiag,\Y+7*\ydiag+6*\yhori);
\U3=(\X+2*\xdiag,\Y+6*\ydiag+6*\yhori);
\U4=(\X+2*\xdiag,\Y+6*\ydiag+5*\yhori);
\p1=(\x,\y);
\p2=(\x-1*\xdiag,\y+\ydiag);
\p3=(\x-1*\xdiag,\y+\ydiag+\yhori);
\p4=(\x-0*\xdiag,\y+2*\ydiag+\yhori);
\p5=(\x+1*\xdiag,\y+1*\ydiag+\yhori);
\p6=(\x+1*\xdiag,\y+1*\ydiag);
\q1=(\x-0*\xdiag,\y+2*\ydiag+2*\yhori);
\q2=(\x+1*\xdiag,\y+3*\ydiag+2*\yhori);
\q3=(\x+2*\xdiag,\y+2*\ydiag+2*\yhori);
\q4=(\x+2*\xdiag,\y+2*\ydiag+1*\yhori);
\r1=(\x-1*\xdiag,\y+3*\ydiag+2*\yhori);
\r2=(\x-1*\xdiag,\y+3*\ydiag+3*\yhori);
\r3=(\x-0*\xdiag,\y+4*\ydiag+3*\yhori);
\r4=(\x+1*\xdiag,\y+3*\ydiag+3*\yhori);
\s1=(\x-0*\xdiag,\y+4*\ydiag+4*\yhori);
\s2=(\x+1*\xdiag,\y+5*\ydiag+4*\yhori);
\s3=(\x+2*\xdiag,\y+4*\ydiag+4*\yhori);
\s4=(\x+2*\xdiag,\y+4*\ydiag+3*\yhori);
\t1=(\x-1*\xdiag,\y+5*\ydiag+4*\yhori);
\t2=(\x-1*\xdiag,\y+5*\ydiag+5*\yhori);
\t3=(\x-0*\xdiag,\y+6*\ydiag+5*\yhori);
\t4=(\x+1*\xdiag,\y+5*\ydiag+5*\yhori);
\u1=(\x-0*\xdiag,\y+6*\ydiag+6*\yhori);
\u2=(\x+1*\xdiag,\y+7*\ydiag+6*\yhori);
\u3=(\x+2*\xdiag,\y+6*\ydiag+6*\yhori);
\u4=(\x+2*\xdiag,\y+6*\ydiag+5*\yhori);
}
\draw (\P1) -- (\P2);
\draw (\P2) -- (\P3);
\draw (\P3) -- (\P4);
\draw (\P4) -- (\P5);
\draw (\P5) -- (\P6);
\draw (\P6) -- (\P1);

\draw (\P4) -- (\Q1);
\draw (\Q1) -- (\Q2);
\draw (\Q2) -- (\Q3);
\draw (\Q3) -- (\Q4);
\draw (\Q4) -- (\P5);

\draw (\Q1) -- (\R1);
\draw (\R1) -- (\R2);
\draw (\R2) -- (\R3);
\draw (\R3) -- (\R4);
\draw (\R4) -- (\Q2);

\draw (\R3) -- (\S1);
\draw (\S1) -- (\S2);
\draw (\S2) -- (\S3);
\draw (\S3) -- (\S4);
\draw (\S4) -- (\R4);

\draw (\S1) -- (\T1);
\draw (\T1) -- (\T2);
\draw (\T2) -- (\T3);
\draw (\T3) -- (\T4);
\draw (\T4) -- (\S2);

\draw (\T3) -- (\U1);
\draw (\U1) -- (\U2);
\draw (\U2) -- (\U3);
\draw (\U3) -- (\U4);
\draw (\U4) -- (\T4);

\draw (\p1) -- (\p2);
\draw (\p2) -- (\p3);
\draw (\p3) -- (\p4);
\draw (\p4) -- (\p5);
\draw (\p5) -- (\p6);
\draw (\p6) -- (\p1);

\draw (\p4) -- (\q1);
\draw (\q1) -- (\q2);
\draw (\q2) -- (\q3);
\draw (\q3) -- (\q4);
\draw (\q4) -- (\p5);

\draw (\q1) -- (\r1);
\draw (\r1) -- (\r2);
\draw (\r2) -- (\r3);
\draw (\r3) -- (\r4);
\draw (\r4) -- (\q2);

\draw (\r3) -- (\s1);
\draw (\s1) -- (\s2);
\draw (\s2) -- (\s3);
\draw (\s3) -- (\s4);
\draw (\s4) -- (\r4);

\draw (\s1) -- (\t1);
\draw (\t1) -- (\t2);
\draw (\t2) -- (\t3);
\draw (\t3) -- (\t4);
\draw (\t4) -- (\s2);

\draw (\t3) -- (\u1);
\draw (\u1) -- (\u2);
\draw (\u2) -- (\u3);
\draw (\u3) -- (\u4);
\draw (\u4) -- (\t4);

\filldraw[black] (\X,{\Y+\ydiag +(1/2)*\yhori}) circle (1.2pt) node{};
\filldraw[black] (\X+\xdiag,{\Y+2*\ydiag +(3/2)*\yhori}) circle (1.2pt) node{};
\filldraw[black] (\X+0*\xdiag,{\Y+3*\ydiag +(5/2)*\yhori}) circle (1.2pt) node{};
\filldraw[black] (\X+1*\xdiag,{\Y+4*\ydiag +(7/2)*\yhori}) circle (1.2pt) node{};
\filldraw[black] (\X+0*\xdiag,{\Y+5*\ydiag +(9/2)*\yhori}) circle (1.2pt) node{};
\filldraw[black] (\X+1*\xdiag,{\Y+6*\ydiag +(11/2)*\yhori}) circle (1.2pt) node{};

\draw ({\X+(1/2)*\xhori},\Y-2*\ydiag) node{$\vdots$} ;
\draw ({\X+(1/2)*\xhori},\Y+18*\ydiag) node{$\vdots$} ;

\draw (\X-4*\xhori,\Y+8*\ydiag) node{$\cdots$} ;
\draw (\X+16*\xhori,\Y+8*\ydiag) node{$\cdots$} ;

\filldraw[black] (\x,{\y+\ydiag +(1/2)*\yhori}) circle (1.2pt) node{};
\filldraw[black] (\x+\xdiag,{\y+2*\ydiag +(3/2)*\yhori}) circle (1.2pt) node{};
\filldraw[black] (\x+0*\xdiag,{\y+3*\ydiag +(5/2)*\yhori}) circle (1.2pt) node{};
\filldraw[black] (\x+1*\xdiag,{\y+4*\ydiag +(7/2)*\yhori}) circle (1.2pt) node{};
\filldraw[black] (\x+0*\xdiag,{\y+5*\ydiag +(9/2)*\yhori}) circle (1.2pt) node{};
\filldraw[black] (\x+1*\xdiag,{\y+6*\ydiag +(11/2)*\yhori}) circle (1.2pt) node{};

\draw ({\x+(1/2)*\xhori},\y-2*\ydiag) node{$\vdots$} ;
\draw ({\x+(1/2)*\xhori},\y+18*\ydiag) node{$\vdots$} ;

\end{tikzpicture}
$$  

There is a magnetically preferred easy axis at an angle of around $\pm 30^{\circ}$ with the Cobalt chain direction, see \cite{HEI}.  As a consequence, Cobalt-Niobate is well modeled by the $1$D quantum Ising model.  The situation becomes very interesting with an externally applied magnetic field orthogonal to all the Cobalt chain easy axes.  For suitable field strength of the external field,  the 1D transverse field quantum Ising model has a quantum phase transition and is described  by a central charge $1/2$ minimal model conformal field theory.  From a Lie algebraic perspective,  this conformal field theory can be described, see \cite{HM},  as a Toda theory:

Let $\mathfrak g$ be a simple complex finite-dimensional Lie algebra of rank $r$, with a real form $\mathfrak g^{*=1}$ of ``Hermitian operators'',  we recall the definition in Section (\ref{Toda-Weyl-definition-section}).  Let $\mathfrak h$ be a Cartan subalgebra of $\mathfrak g$ and consider a field $\phi : \mathbb{R}^{2} \rightarrow \mathfrak h \cap \mathfrak g^{*=1}$.  Let $R$ be a subset of the roots.  Let $(-,-)$ denote the Killing form and define the Lagrangian
$$\textrm{L}_{R}= \frac{1}{2}(\partial_{\mu} \phi, \partial^{\mu} \phi)  - \sum_{\alpha_i \in R} \exp(\alpha_i \cdot \phi)$$
If $R=\{\alpha_1,\cdots,\alpha_r\}$ is chosen to be a set of simple roots,  one obtains the usual conformal and massless Toda theory (we have set all coupling constants to $1$,  for simplicity).  Its quantization, as carried out in \cite{HM}, describes for $\mathfrak g= \mathfrak s \mathfrak l_{2}$ and $\mathfrak g= \mathfrak e_{8}$ the previously mentioned minimal model. 

The actual magnetic structure of $\textrm{Co} \textrm{Nb}_{2} \textrm{O}_6$ incorporates a non-zero, yet small,  interchain coupling.  One obtains a Hamiltonian with an additional term corresponding to magnetization along the easy axis.  This magnetic deformation of the transverse field Ising model has a beautiful Lie algebraic formulation on the level of Toda theory: Replace the set $R$ of simple roots by $R_{\textrm{new}} = R \cup \{ \alpha_{0}\}$ for the lowest root $\alpha_0$.   The resulting affine Toda theory acquires a fascinating mass spectrum, that is known to correspond to the entries of a Perron-Frobenius eigenvector of the Cartan matrix of $\mathfrak g$, see \cite{FLO}. Since in the case $\mathfrak g=\mathfrak e_{8}$ affine Toda theory describes the quantum phase transition of $\textrm{Co} \textrm{Nb}_{2} \textrm{O}_6$ in the presence of a small interchain coupling, this confirms the mass spectrum predicted in groundbreaking work by Zamolodchikov \cite{ZAM}.  Coldea et al.  in \cite{COL} were able to confirm experimentally for $\textrm{Co} \textrm{Nb}_{2} \textrm{O}_6$ some of these predictions using neutron diffraction, in particular that the ratio of second lightest to lightest mass is given by the golden ratio.  As we are about to recall,  affine Toda theory is intimately linked to the Coxeter element of the Weyl group of $\mathfrak g$. The aim of the present work is to generalize these mass spectrum calculations for Lagrangians based on much more general Weyl group elements.  We show that the classical masses can again be expressed in terms of the relative geometry of a special collection of roots,  generalizing the mass expression in terms of the Cartan matrix.

The starting point is Freeman's slight reformulation in \cite{FRE} of affine Toda theory. For $R_{\textrm{new}}$ as before,  rewrite the field as $\phi_{\textrm{min}} + \phi$ where $\phi_{\textrm{min}}$ is a minimum.  The resulting Lagrangian for the new field $\phi$ is  
\begin{eqnarray}
\label{new-Lagrangian}
 \frac{1}{2}(\partial_{\mu} \phi, \partial^{\mu} \phi) - \sum_{\alpha_i \in R_{\textrm{new}}}  n_i \exp(\alpha_i \cdot \phi)
\end{eqnarray}
where $n_0$ is normalized to be $1$ and the $n_i$'s are such that 
\begin{eqnarray}
\label{linear-relation-equation}
\alpha_0=- \sum_{i=1}^{r} n_i \alpha_i
\end{eqnarray} 
See for example \cite{BCD} (Section 2) for details.  Using the Killing form, identify roots as elements of the Cartan subalgebra $\mathfrak h$ of $\mathfrak g$.  For each root $\alpha$ fix a generator $\textrm{E}_{\alpha}$ of the corresponding root space with respect to $\mathfrak h$,  normalized so that 
$$[E_{\alpha},E_{-\alpha}]=\alpha \in \mathfrak h$$
Let 
\begin{align*}
\Lambda_{+} &= \sum_{i=0}^{r} a_i \textrm{E}_{\alpha_i} \\[0.07in]
\Lambda_{-} &=\sum_{i=0}^{r} b_i \textrm{E}_{-\alpha_i}
\end{align*}
where $a_i b_i = n_i$ for each $i$.  Since $\alpha_{0}$ is the lowest root,  $[E_{\alpha_{i}},E_{-\alpha_{j}}]=\delta_{i,j} \alpha_{i}$ for all $0\le i,j \le r$.  Hence, the Lagrangian in Equation (\ref{new-Lagrangian}) can be re-written as
\begin{eqnarray}
\label{crucial-Lagrangian}
 \frac{1}{2}(\partial_{\mu} \phi, \partial^{\mu} \phi)  - (\exp(\textrm{ad } \phi )(\Lambda_{+}) , \Lambda_{-})
\end{eqnarray}
One can show,  see the work of Kostant \cite{KOS} (Section 6),  that the centralizer of $\Lambda_{+}$ is a Cartan algebra,  denote it by $\mathfrak h^{\textrm{Weyl}}$.  It follows from Equation (\ref{linear-relation-equation}) that $\Lambda_{-}$ is in $\mathfrak h^{\textrm{Weyl}}$. The Weyl group $W$ of $\mathfrak g$ acts on this Cartan algebra and Kostant has shown in loc. cit. that there is a Coxeter element $\sigma_{\textrm{Coxeter}}$ in $W$ such that 
$$\sigma_{\textrm{Coxeter}}(\Lambda_{\pm}) = e^{\pm \frac{2 \pi i}{h}} \cdot  \Lambda_{\pm} $$
where $h$ is the Coxeter number.  In this manner, Equation (\ref{crucial-Lagrangian}) allows to formulate affine Toda theory in terms of the linear algebra of Coxeter elements.  We extend this formalism to more general Weyl group elements and calculate the classical mass spectrum.  

For usual affine Toda theory, the masses were calculated in the early 90's by Dorey \cite{DOR}, Freeman \cite{FRE}, and Fring-Liao-Olive \cite{FLO}: They can beautifully be expressed in terms of the Perron-Frobenius eigenvector of the Cartan matrix of $\mathfrak g$.  More recently,  Brillon-Schechtman in  \cite{BS} considered the case of the Coxeter element but with $\Lambda_{+}$ replaced by an arbitrary eigenvector.  Many other works concern various generalizations of Toda theory,  see for example \cite{DGPZ},  \cite{FK1}, \cite{FK2}, \cite{FW}, \cite{GMT}.

\subsection{The Toda-Weyl Lagrangians}
\label{Toda-Weyl-definition-section}

In this section we generalize the affine Toda Lagrangians described by Equation (\ref{crucial-Lagrangian}).  As before,  let $\mathfrak g$ be a simple finite-dimensional complex Lie algebra of rank $r$,  and fix a Cartan subalgebra $\mathfrak h^{\textrm{Weyl}}$.  Let $(-,-)$ denote the Killing form and identify elements in root space as elements of $\mathfrak h^{\textrm{Weyl}}$ by associating to $x$ in $\mathfrak h^{\textrm{Weyl}}$ the functional $(x,-)$.  For every root $\alpha$ consider the endomorphism of $\mathfrak h^{\textrm{Weyl}}$ given by
$$r_{\alpha}(\beta)= \beta - 2 \cdot \frac{(\alpha,\beta)}{(\alpha,\alpha)} \cdot  \alpha$$

The Weyl group $W$ of $\mathfrak g$ is the group generated by the $r_{\alpha}$'s.  For $\sigma$ in $W$ we aim to construct an analogue of the Lagrangian in Equation (\ref{crucial-Lagrangian}). This requires two things: Define the target space of the field $\phi$,  and generalize the special elements $\Lambda_{\pm}$ in the Lie algebra $\mathfrak g$.  

 For each root $\alpha$ choose generators $e_{\pm 
\alpha}$ of the root space of $\pm \alpha$ (with respect to  $\mathfrak h^{\textrm{Weyl}}$), normalized so that $[e_{\alpha},e_{-\alpha}]=
\alpha$.  Let $\mathfrak a$ be the real subalgebra of $\mathfrak g$ given by the 
$\mathbb{R}$-span of the following elements, as $\alpha$ ranges through the set of all roots:
\begin{enumerate}[(i)]
\item
$\alpha$
\item
$e_{\alpha}+e_{-\alpha}$
\item
$i(e_{\alpha}-e_{-\alpha})$
\end{enumerate}
Consider the involution $* : \mathfrak g \rightarrow \mathfrak g$ discussed in detail by Kostant \cite{KOS} (Section 6): It is given by
$$*  : x+ i y \mapsto x-iy$$
where $x$ and $y$ are in $\mathfrak a$. Then $\mathfrak a = \mathfrak g^{*=1}$ can be considered the space of Hermitian operators, it is a real form of the complex algebra $\mathfrak g$. The field $\phi$ will take values in a real subspace of $\mathfrak g^{*=1}$.  To define this subspace, we appeal to results by Kac.

Starting with $\sigma$ and the Cartan algebra $\mathfrak h^{\textrm{Weyl}}$,  we recall how to construct a second Cartan algebra $\mathfrak h^{\textrm{Kac}}$ as well as a gradation $\mathfrak g = \oplus_{k} \mathfrak g_{k}$ (if $\sigma$ is a Coxeter element, the two Cartan algebras $\mathfrak h^{\textrm{Kac}}$ and $\mathfrak h^{\textrm{Weyl}}$ are in apposition, in the terminology of Kostant \cite{KOS}).  Kac showed,  see \cite{KAC} (Theorem 8.6),  that there is a Cartan algebra $\mathfrak h^{\textrm{Kac}}$ and a collection of non-negative integers $\textbf{s}=(s_0,\cdots,s_r)$ with the following properties: 
\begin{enumerate}[(i)]
\item
There is a finite-order inner automorphism $\tilde \sigma$ of $\mathfrak g$ (of order $\tilde n$,  say) such that
\begin{align*}
\tilde \sigma |_{\mathfrak h^{\textrm{Weyl}}} &=\sigma \\[0.1in]
\tilde \sigma |_{\mathfrak h^{\textrm{Kac}}} &= \textrm{id}
\end{align*}
\item 
Let $\zeta =e^{\frac{2\pi i}{\tilde n}}$,  let $\alpha_1,\cdots,\alpha_r$ be a set of simple roots, and let $\alpha_0$ denote the lowest root.  For all $j=0,\cdots,r$
$$\tilde \sigma (\textrm{E}_{\alpha_{j}}) =\zeta^{s_{j}} \cdot \textrm{E}_{\alpha_{j}}$$
where  $\textrm{E}_{\alpha_{j}}$ generates the $\alpha_{j}$ root space with respect to $\mathfrak h^{\textrm{Kac}}$.
\item
The $s_j$'s that are non-zero are co-prime.
\end{enumerate}
The collection of integers $\textbf{s}=(s_0,\cdots,s_r)$ are called Kac coordinates of $\sigma$.  In general,  there are multiple possible coordinates $\textbf{s}$ associated to a given Weyl group element.  In some situations however, for example if $\sigma$ lies in a regular conjugacy class,  they are uniquely defined.  See for example \cite{REE} for more details.  Associate to $\tilde \sigma$ a $\mathbb{Z}/ \tilde n\mathbb{Z}$-gradation on $\mathfrak g$, by letting $\mathfrak g_{k}$ denote the $\tilde \sigma$-eigenspace  with eigenvalue $e^{2\pi i k/\tilde n}$.  The space $\mathfrak g_0$ then contains $\mathfrak h^{\textrm{Kac}}$ but is in general larger.  We can now define the target space of the field $\phi$: We assume 
$$\phi : \mathbb{R}^{2} \longrightarrow \mathfrak g_{0} \cap \mathfrak g^{*=1}$$
Furthermore,  let from now on $\Lambda_{+}$ in $\mathfrak h^{\textrm{Weyl}}$ denote an eigenvector of $\sigma$ with eigenvalue $\mu$, say,  and let $\Lambda_{-}=*(\Lambda_{+})$.  It follows from the definition of $*$ that it maps $\mathfrak h^{\textrm{Weyl}}$ to itself. Hence $\Lambda_{-}$ is again in $\mathfrak h^{\textrm{Weyl}}$ and is in fact an eigenvector of $\sigma$ with complex conjugate eigenvalue compared to $\Lambda_{+}$.  To summarize the situation:

$$\xymatrix{\mathbb{R}^{2} \ar^-{\phi}[r] & \mathfrak g_{0} \cap \mathfrak g^{*=1} \; \ar@{^{(}->}[r] &\mathfrak g_{0} &&& & \sigma (\Lambda_{\pm})  = \mu^{\pm 1} \Lambda_{\pm}  &   \\
&&\mathfrak h^{\textrm{Kac}}   \; \ar@{^{(}->}[]!<0ex,3ex>;[u]!<0ex,0ex>
\ar@{<.>}_-{\textrm{generalized apposition}}[rrrr] && &&  \; \mathfrak h^{\textrm{Weyl}} \supset \{ \Lambda_{+}, \Lambda_{-} \} &
}$$

\begin{defi}
For $\phi$, $\Lambda_{+}$, and $\Lambda_{-}$ as above,  define the Lagrangian exactly as in Equation (\ref{crucial-Lagrangian}) by
\begin{align*}
\textrm{L}_{\textrm{Toda-Weyl}}=  \frac{1}{2}(\partial_{\mu} \phi, \partial^{\mu} \phi)  - (\exp(\textrm{ad } \phi )(\Lambda_{+}) , \Lambda_{-})
\end{align*}
We call this the Toda-Weyl theory.
\end{defi}
Suppose $\sigma$ is a Coxeter element and $\Lambda_{+}$ has eigenvalue $e^{2\pi i /h}$.  The Kac coordinates are known to be $(1,\cdots,1)$ and hence $\mathfrak g_{0} = \mathfrak h_{\textrm{Kac}}$.  If the Cartan algebra $\mathfrak h$ in Section \ref{introduction-section} is identified as $\mathfrak h_{\textrm{Kac}}$,  the Toda-Weyl theory recovers affine Toda theory.

\section{Mass calculations}

For $1\le i \le s$ consider complex-valued scalar fields $\phi_i : \mathbb{R}^{2} \rightarrow \mathbb{C}$.  Let $\phi=(\phi_{1},\cdots,\phi_{s})^{\textrm{T}}$ and consider a Lagrangian density
$$\textrm{L}= \frac{1}{2} \cdot \partial_{\mu} \overline{\phi}^{\textrm{T}}  A  \partial_{\mu}  \phi  + B \phi + \overline{\phi}^{\textrm{T}} C  \phi + 
\cdots $$
where $A$ and $C$ are $s\times s$ matrices and $B$ is of the form  $(b_1,\cdots,b_s)$. Under a field redefinition $\phi \mapsto X \phi$ (with $X$ invertible) one obtains
$$\textrm{L} \mapsto \textrm{L}_{X}= \frac{1}{2} \cdot \partial_{\mu} \overline{\phi}^{\textrm{T}}  (\overline{X}^{\textrm{T}} A X)  \partial_{\mu} \phi  + BX \phi + \overline{\phi}^{\textrm{T}} (\overline{X}^{\textrm{T}} C X)  \phi + 
\cdots $$
Suppose that for some $X$
\begin{eqnarray}
\label{diagonal-Lagrangian}
 \textrm{L}_{X}= \frac{1}{2} \cdot \partial_{\mu} \overline{ \phi}^{\textrm{T}}  \; \partial_{\mu} \phi  - \frac{1}{2} \cdot  \overline{\phi}^{\textrm{T}} \begin{bmatrix}
m_{1}^{2} &&&\\
&m_{2}^{2} &&\\
&&\ddots&\\
&&&m_{s}^{2}
\end{bmatrix}  \phi + \cdots 
\end{eqnarray}
where the $m_{i}$'s are non-negative real numbers. The masses of the Lagrangian are then by definition $m_1,\cdots, m_s$.  If the masses exist,  they are well defined: If $Y$ is another field redefinition such that $\textrm{L}_{Y}$ is as in Equation (\ref{diagonal-Lagrangian}), then $Y=XU$ for a unitary matrix $U$.  Hence the spectrum of $\overline{X}^{\textrm{T}} C X$ agrees with the spectrum of $\overline{Y}^{\textrm{T}} C Y$.  Furthermore, the linear $\phi$ term in $\textrm{L}_{X}$ vanishes if and only if $B=0$ and hence vanishes in $\textrm{L}_{Y}$ as well.

In the remainder of this work we calculate the masses of the Toda-Weyl Lagrangians.  The key idea is to construct a basis of $\mathfrak g_0$ in terms of the root space decomposition with respect to $\mathfrak h^{\textrm{Weyl}}$ (for Coxeter elements such a description goes back to work of Kostant \cite{KOS}).  This allows a description of the masses in terms of pairings $\Lambda_{+} \cdot \alpha := \alpha(\Lambda_{+})$ for suitable roots $\alpha$, viewed as functionals on $\mathfrak h^{\textrm{Weyl}}$. 
 
\begin{thm}
\label{main-theorem}
Let $\mathfrak g$ be a simple finite-dimensional complex Lie algebra with a Cartan subalgebra $\mathfrak h^{\textrm{\emph{Weyl}}}$.  Let $\sigma$ be an element in the Weyl group of $\mathfrak g$ such that:
\begin{enumerate}[{\normalfont (i)}]
\item
$1$ is not an eigenvalue of $\sigma$ acting on $\mathfrak h^{\textrm{\emph{Weyl}}}$.
\item
There is an inner automorphism $\tilde \sigma$ of $\mathfrak g$ such that $\tilde \sigma |_{\mathfrak h^{\textrm{\emph{Weyl}}}} = \sigma$ and $\textrm{\emph{ord }} \tilde \sigma = \textrm{\emph{ord }}  \sigma$. 
\end{enumerate}
Let $\Lambda_{+} \in \mathfrak h^{\textrm{\emph{Weyl}}}$ be an eigenvector of $\sigma$ and let $\Lambda_{-} = *(\Lambda_{+})$.  Let  $\gamma_1,\cdots, \gamma_{s}$ denote orbit representatives for the action of the cyclic group $\langle \sigma \rangle$ on the set of roots.  Then the masses of
$$\textrm{\emph{L}}=\frac{1}{2}(\partial_{\mu} \phi, \partial^{\mu} \phi)  - (\exp(\textrm{\emph{ad }} \phi )(\Lambda_{+}) , \Lambda_{-})$$
are given by
$$m_i = | \Lambda_{+} \cdot \gamma_i|$$
\end{thm}
\begin{proof}
Let us show that the absolute values $|\Lambda_{+} \cdot \gamma_{i}|$ are independent of the choice of orbit representatives: Let $\tilde \Lambda_{+} = (\Lambda_{+},-)$ be the element in root space associated to $\Lambda_{+}$.  Recall that $\Lambda_{+}$ is an eigenvector of $\sigma$,  denote the corresponding eigenvalue $\mu$, a suitable root of unity.  It follows that for all $j$
\begin{align*}
|\Lambda_{+} \cdot \sigma^{j} \gamma_{i}|&=| (\sigma^{j} \gamma_{i})(\Lambda_{+})| \\[0.07in]
&=|(\sigma^{j} \gamma_{i}, \tilde \Lambda_{+})|\\[0.07in]
&=|( \gamma_{i},  \sigma^{-j} \tilde \Lambda_{+})| \\[0.07in]
&=|\mu^{-j} ( \gamma_{i}, \tilde \Lambda_{+})| \\[0.07in]
&= |\Lambda_{+} \cdot \gamma_{i}|
\end{align*}
where we have used that the inner product in root space is invariant under the Weyl group action. 

Let $\tilde \sigma$ be a finite-order inner automorphism lifting $\sigma$.  For each root $\alpha$ let $e_{\alpha}$ be a generator of the root space with respect to $\mathfrak h_{\textrm{Weyl}}$. Then $\tilde \sigma e_{\alpha}$ is a generator of the root space of $\sigma(\alpha)$.  By the assumptions of the theorem, we can assume that $\tilde \sigma$ is of the same order as $\sigma$. Hence, one can choose the generators of the root spaces such that for each root $\alpha$ 
\begin{eqnarray}
\label{sigma-equation}
\tilde \sigma e_{\alpha} =e_{\sigma(\alpha)}
\end{eqnarray}
For each orbit $\mathcal O_{i}$ pick a representative $\gamma_i$ and define
$$A_{i}:= \frac{1}{\sqrt{|\mathcal O_{i}|}} \cdot \sum_{j=0}^{|\mathcal O_{i}|-1} e_{\sigma^{j}(\gamma_{i})}$$
From Equation (\ref{sigma-equation}) it follows that $A_{i}$ is fixed by $\tilde \sigma$ and hence lies in $\mathfrak g_0$.  Since the root space generators are linearly independent, so are the $A_{i}$'s.  We have seen that $\tilde \sigma$ permutes the various root spaces, and since $\tilde \sigma|_{\mathfrak h^{\textrm{Weyl}}} = \sigma$ it follows that $\mathfrak g_{0}$ has a basis given by the $A_i$'s together with a basis of $(\mathfrak h^{\textrm{Weyl}})^{\sigma=1}$. By our assumption on $\sigma$,  $(\mathfrak h^{\textrm{Weyl}})^{\sigma=1}=0$ and therefore the $A_{i}$'s are in fact a basis of $\mathfrak g_0$. We now show that in this basis one can read off the masses of the Toda-Weyl Lagrangian.  Write
$$\phi = \sum_{i=1}^{s} \phi_{i} A_i$$
Since $(e_{\alpha},e_{\beta})=0$ unless $\alpha+\beta=0$, one can deduce that
\begin{eqnarray}
\label{AiAj-equation}
(A_i,A_j) = \delta_{i,\pi(j)}
\end{eqnarray}
where $\pi$ is a permutation of the indices such that $-\gamma_{j}$ is in $A_{\pi(j)}$ for all $j$. By \cite{BS} (Theorem 2.4)
\begin{eqnarray*}
[[ A_j , \Lambda_{+}],\Lambda_{-}]&=& (\Lambda_{+}\cdot \gamma_{j}) \cdot (\Lambda_{-} \cdot \gamma_{j}) \cdot A_j
\end{eqnarray*}
Note that by \cite{KOS} (Equation 6.1.1),  for every complex scalar $c$ 
\begin{align}
\label{complex-conjugate-equation}
(c \cdot e_{\alpha})^{*}= \overline{c} \cdot e_{-\alpha}
\end{align}
It follows from the construction of $*$ that for all $x,y$ in $\mathfrak g$
$$(*(x),*(y))= \overline{(x,y)}$$
In particular:
\begin{align*}
\Lambda_{-}\cdot \gamma_{j} &=  *(\Lambda_{+}) \cdot *(\gamma_{j}) \\[0.07in]
&=  \overline{\Lambda_{+} \cdot \gamma_{j}}
\end{align*}
Therefore
\begin{align*}
([ A_i,[A_j , \Lambda_{+}]],\Lambda_{-})&= (A_{i}, [[ A_j , \Lambda_{+}],\Lambda_{-}] ) \\[0.07in]
&=  (\Lambda_{+}\cdot \gamma_{j})\cdot ( \Lambda_{-}\cdot \gamma_{j})(A_i, A_j)\\[0.07in]
&= |\Lambda_{+}\cdot \gamma_{i}|^{2} \cdot \delta_{i,\pi(j)}
\end{align*}
Together with  Equation (\ref{AiAj-equation}),  this implies that the Toda-Weyl Lagrangian density is given by
\begin{align*}
\textrm{L} &=  \frac{1}{2}\cdot \sum_{i=1}^{s} \partial_{\mu} \phi_{i} \partial^{\mu} \phi_{\pi(i)}- (\Lambda_{+},\Lambda_{-}) - \sum_{i=1}^{s} \phi_{i} ([A_{i}, \Lambda_{+}],\Lambda_{-}) - \frac{1}{2!} \sum_{i,j} \phi_{i} \phi_{j} ([ A_i,[A_j , \Lambda_{+}]],\Lambda_{-}) + \textrm{higher order terms}   \\[0.07in]
&=  \frac{1}{2}\cdot \sum_{i} \partial_{\mu} \phi_{i} \partial^{\mu} \phi_{\pi(i)}- (\Lambda_{+},\Lambda_{-}) - \sum_{i=1}^{s} \phi_{i} (A_{i}, [\Lambda_{+},\Lambda_{-}])-\frac{1}{2!} \sum_{i,j} \phi_{i} \phi_{j} |\Lambda_{+}\cdot \gamma_{i}|^{2} \cdot \delta_{i,\pi(j)} + \textrm{higher order terms} 
\end{align*}
To simplify further,  note that by construction the image of $\phi$ is contained in the space fixed by $*$.  Hence
\begin{align*}
\sum_{i=1}^{s} \phi_{i} A_{i} & = * \left (\sum_{i=1}^{s} \phi_{i} A_{i}  \right ) \\[0.07in]
&=\sum_{i=1}^{s} \overline{\phi_{i}} A_{\pi(i)}
\end{align*}
and therefore
$$\phi_{\pi(i)}= \overline{\phi_{i}}$$
for all $i$.  Furthermore $[\Lambda_{+},\Lambda_{-}]=0$.  Therefore,  up to the constant $(\Lambda_{+},\Lambda_{-})$,  the Lagrangian density is given by
$$\textrm{L}= \frac{1}{2}\cdot \sum_{i} \partial_{\mu} \phi_{i} \partial^{\mu} \overline{\phi}_{i}- \frac{1}{2}  \sum_{i} |\phi_{i}|^{2}  |\Lambda_{+}\cdot \gamma_{i}|^{2} + \textrm{ higher order terms} $$
This implies the theorem. 
\end{proof}

\begin{rem}
\label{condition-remark}
The number of Toda-Weyl masses often has a very simple expression. Suppose for example $\sigma$ is a regular Weyl group element such that $1$ is not an eigenvalue.  As discussed by Reeder in \cite{REE} (Proposition 2.2),  since $\sigma$ is regular there indeed exists a lift $\tilde \sigma$ of the same order as $\sigma$.  Hence the conditions of Theorem \ref{main-theorem} are satisfied.  Since $\sigma$ is regular,  it follows from the work of Springer \cite{SPR} (Proposition 4.1) that every orbit has exactly $n=\textrm{ord } \sigma$ elements.  The total number of roots is known to be $hr$ where $h$ is the Coxeter number and $r$ is the rank of $\mathfrak g$.  It follows that the number $s$ of orbits of $\sigma$ satisfies
$$\textrm{ ord } \sigma \cdot s=h \cdot \textrm{rank } \mathfrak g$$
In particular, the number of Toda-Weyl masses is given by 
$$s=\frac{h \cdot \textrm{rank } \mathfrak g}{\textrm{ ord } \sigma}$$
\end{rem}

In Section \ref{first-example} and Section \ref{second-example} we give two illustrative examples of how to use Theorem \ref{main-theorem} to obtain the precise mass spectrum.

\subsection{Example I}
\label{first-example} 

Let $\mathfrak g= \mathfrak e_6$, with simple roots $\alpha_1,\cdots,\alpha_6$ indexed as in \cite{BOURBAKI}.  Consider the Weyl group element
$$\sigma= r_{\alpha_1}r_{\alpha_2}r_{\alpha_5} r_{\alpha_6} r_{\alpha_2+\alpha_4}r_{\alpha_3+\alpha_4} $$
Since the $6$ roots involved in the definition of $\sigma$ are a basis of root space, it follows from work of Carter \cite{CAR} (Lemma 3) that $1$ is not an eigenvalue of $\sigma$.  Furthermore, one can calculate that the eigenvalues are distinct: For $\zeta_{9}=e^{2\pi i/9}$, they are $\zeta_{9},\zeta_{9}^{2},\zeta_{9}^{4},\zeta_{9}^{5}, \zeta_{9}^{7}, \zeta_{9}^{8}$, see for example \cite{BOU} (Table 1).  Hence, by work of Springer \cite{SPR} (Lemma 4.11), the element $\sigma$ is regular. By Remark \ref{condition-remark} it follows that the conditions of Theorem \ref{main-theorem} are satisfied.  We apply the theorem with $\Lambda_{+}$ an eigenvector with eigenvalue $\zeta_{9}$.  Since $\zeta_{9}$ has multiplicity $1$,  up to an overall scaling, the masses are independent of $\Lambda_{+}$.
 
Let $\zeta=e^{2\pi i/36}$ be a primitive $36$'th root of unity.  It has minimal polynomial over $\mathbb{Q}$ given by $x^{12}-x^{6}+1$.  We claim that one can take $\Lambda_{+}$ as 
\begin{align}
\label{exampleI-Lambda}
 \zeta \alpha_1 +(-\zeta^{7}+\zeta^{5}+\zeta)\alpha_2+(-\zeta^{9}+\zeta^3+\zeta)\alpha_3+(-\zeta^{11}-\zeta^{7}+\zeta^{5}+\zeta^{3}+\zeta)\alpha_4 
+(-\zeta^{11}+\zeta^{5}+\zeta^{3})\alpha_{5}+(-\zeta^{11}+\zeta^{5})\alpha_6
\end{align}
To show this, note that the $8$ orbits $\mathcal O_{1},\cdots, \mathcal O_{8}$ of the cyclic group $\langle \sigma \rangle$ acting on the set of roots can be calculated easily:
 
\hspace{0.2in}
  
\def\arraystretch{1.3}

\begin{center}

\begin{tabular}{p{3.3cm}|p{2.5cm}|p{4.8cm}|p{4.1cm}}  
$\mathcal O_{1}$ &$ \mathcal O_{2} $&$ \mathcal O_{3} $&$ \mathcal O_{4}$\\
\hline   
\hline 
&&& \\
$-\alpha_{6}$ &$ - \alpha_{2}-\alpha_{4} $&$ \alpha_{3}-\alpha_{4}$&$ \alpha_{5}$\\

$\alpha_{5}+\alpha_{6} $&$\alpha_{4}+\alpha_{5} $&$ \alpha_{1}+\cdots + \alpha_{5}$&$\alpha_{1}+\alpha_{2}+\alpha_{3}+2\alpha_{4}+2\alpha_{5}+\alpha_{6}$ \\

$\alpha_{1}+\alpha_{2}+\alpha_{3}+2\alpha_{4}+\alpha_{5}$ &$\alpha_{2}+\alpha_{4}+\alpha_{5}+\alpha_{6} $&$ \alpha_{1}+\alpha_{2}+2\alpha_{3}+2\alpha_{4}+2\alpha_{5}+\alpha_{6} $&$ \alpha_{1}+2\alpha_{2}+2\alpha_{3}+3\alpha_{4}+2\alpha_{5}+\alpha_{6}$\\

$\alpha_{2}+\cdots + \alpha_{6}$ &$\alpha_{1}+\alpha_{2}+\alpha_{3}+\alpha_{4} $&$ \alpha_{1}+\alpha_{2}+2\alpha_{3}+3\alpha_{4}+2\alpha_{5}+\alpha_{6} $&$\alpha_{1}+\alpha_{2}+2\alpha_{3}+2\alpha_{4}+\alpha_{5}+\alpha_{6}$ \\

$\alpha_{1}+\alpha_{3}+\alpha_{4} $&$\alpha_{3} $& $\alpha_{2}+\alpha_{3}+2\alpha_{4}+\alpha_{5}+\alpha_{6} $&$\alpha_{3}+\alpha_{4}$\\

$-\alpha_{1}$ &$-\alpha_{2} $&$ -\alpha_{5} $&$-\alpha_{1}-\cdots-\alpha_{5}$ \\

$-\alpha_{2}-\alpha_{3}-\alpha_{4}-\alpha_{5} $&$-\alpha_{1}-\alpha_{3}$&$  -\alpha_{1}-\alpha_{2}-\alpha_{3}-2\alpha_{4}-2\alpha_{5}-\alpha_{6} $&$ -\alpha_{1}-\alpha_{2}-2\alpha_{3}-2\alpha_{4}-2\alpha_{5}-\alpha_{6}$\\

$-\alpha_{1}-\alpha_{3}-\alpha_{4}-\alpha_{5}-\alpha_{6} $&$ -\alpha_{3}-\alpha_{4}-\alpha_{5} $&$ -\alpha_{1}-2\alpha_{2}-2\alpha_{3}-3\alpha_{4}-2\alpha_{5}-\alpha_{6}$&$-\alpha_{1}-\alpha_{2}-2\alpha_{3}-3\alpha_{4}-2\alpha_{5}-\alpha_{6}$\\

$-\alpha_{2}-\alpha_{3}-2\alpha_{4}-\alpha_{5}$ &$ -\alpha_{4}-\alpha_{5}-\alpha_{6}$ &$ -\alpha_{1}-\alpha_{2}-2\alpha_{3}-2\alpha_{4}-\alpha_{5}-\alpha_{6} $&$-\alpha_{2}-\alpha_{3}-2\alpha_{4}-\alpha_{5}-\alpha_{6}$\\
&&&
\end{tabular}

\end{center}

\begin{center}
 
\begin{tabular}{p{3.3cm}|p{2.5cm}|p{4.8cm}|p{4.1cm}}  
$\mathcal O_{5} $&$ \mathcal O_{6}$ &$ \mathcal O_{7} $&$ \mathcal O_{8}$\\
\hline   
\hline 
&&& \\
$\alpha_{2}$&$ \alpha_{1}$&$ \alpha_{4}$ &$\alpha_{1}+\alpha_{3}+\alpha_{4}+\alpha_{5}$\\

$\alpha_{1}+\alpha_{3} $&$\alpha_{2}+\cdots+\alpha_{5} $&$-\alpha_{1}-\alpha_{3}-\alpha_{4}-\alpha_{5} $&$ \alpha_{2}+\alpha_{3}+2\alpha_{4}+2\alpha_{5}+\alpha_{6}$\\

$\alpha_{3}+\alpha_{4}+\alpha_{5}$ &$ \alpha_{1}+\alpha_{3}+\cdots+\alpha_{6} $&$ -\alpha_{2}-\alpha_{3}-2\alpha_{4}-2\alpha_{5}-\alpha_{6}$&$\alpha_{1}+\alpha_{2}+\alpha_{3}+2\alpha_{4}+\alpha_{5}+\alpha_{6}$\\

$\alpha_{4}+\alpha_{5}+\alpha_{6}$&$\alpha_{2}+\alpha_{3}+2\alpha_{4}+\alpha_{5} $&$ -\alpha_{1}-\alpha_{2}-\alpha_{3}-2\alpha_{4}-\alpha_{5}-\alpha_{6} $&$\alpha_{2}+\alpha_{3}+\alpha_{4}$\\

$\alpha_{2}+\alpha_{4} $&$ \alpha_{6}$&$-\alpha_{2}-\alpha_{3}-\alpha_{4}$&$-\alpha_{2}-\alpha_{4}-\alpha_{5}$\\

$-\alpha_{4}-\alpha_{5} $&$-\alpha_{5}-\alpha_{6} $&$\alpha_{2}+\alpha_{4}+\alpha_{5} $&$-\alpha_{1}-\cdots-\alpha_{6}$\\

$-\alpha_{2}-\alpha_{4}-\alpha_{5}-\alpha_{6} $&$ -\alpha_{1}-\alpha_{2}-\alpha_{3}-2\alpha_{4}-\alpha_{5} $&$ \alpha_{1}+\cdots+ \alpha_{6}$&$-\alpha_{1}-\alpha_{2}-2\alpha_{3}-2\alpha_{4}-\alpha_{5}$\\

$-\alpha_{1}-\alpha_{2}-\alpha_{3}-\alpha_{4} $&$- \alpha_{2}-\cdots -\alpha_{6}$ &$\alpha_{1}+\alpha_{2}+2\alpha_{3}+2\alpha_{4}+\alpha_{5}$&$-\alpha_{3}-\alpha_{4}-\alpha_{5}-\alpha_{6}$\\

$-\alpha_{3}$ &$ -\alpha_{1}-\alpha_{3}-\alpha_{4} $&$\alpha_{3}+\alpha_{4}+\alpha_{5}+\alpha_{6} $&$-\alpha_{4}$
\end{tabular}

\end{center}

\def\arraystretch{1.0}

\hspace{0.2in}

In particular 
\begin{eqnarray*}
\sigma(\alpha_1) &=&\alpha_2 + \alpha_3 + \alpha_4 +\alpha_5\\
\sigma(\alpha_2) &=& \alpha_1+\alpha_3\\
\sigma(\alpha_3)&=&-\alpha_2\\
\sigma(\alpha_4)&=&-\alpha_{1}-\alpha_{3}-\alpha_{4}-\alpha_{5}\\
\sigma(\alpha_5) &=& \alpha_1+\alpha_2+ \alpha_3+2\alpha_4+2\alpha_5 + \alpha_6\\
\sigma(\alpha_6) &=& -\alpha_5-\alpha_6
\end{eqnarray*}
It follows that
$$\sigma(\Lambda_{+}) = \zeta^{4} \cdot \Lambda_{+}=\zeta_{9} \cdot \Lambda_{+}$$
and $\Lambda_{+}$ is indeed the desired eigenvector.  Orbit representatives for the $\sigma$ action on the set of roots are given by 
$$\alpha_1 \; , \;  \alpha_2 \; , \; \alpha_3\; ,  \; \alpha_4 \; ,  \; \alpha_5 \; ,  \; -\alpha_1 \; ,  \; -\alpha_4 \; ,  \; -\alpha_5$$
By Theorem \ref{main-theorem} one obtains the following $8$ masses:
\begin{align*}
|\Lambda_{+} \cdot \pm \alpha_1 |& =|\zeta^{9}-\zeta^3+\zeta|   \approx   0.684\\[0.07in]
|\Lambda_{+} \cdot \alpha_2| &  =|\zeta^{11}-\zeta^{7}+\zeta^5-\zeta^3+\zeta|  \approx  0.446\\[0.07in]
|\Lambda_{+} \cdot \alpha_3| & =|\zeta^{11}-2\zeta^{9}+\zeta^{7}-\zeta^{5}+\zeta^3|  \approx  0.446\\[0.07in]
|\Lambda_{+} \cdot \pm \alpha_4| & =|-\zeta^{11}+\zeta^{9}-\zeta^7| \approx  0.879\\[0.07in]
|\Lambda_{+} \cdot \pm \alpha_5| & =|\zeta^{7}+\zeta^{3}-\zeta|  \approx  1.286
\end{align*}
Let us normalize the masses so that the lowest mass equals $1$ and let us index these normalized masses as 
$$m_1\le \cdots \le m_{8}$$
In the following calculations we repeatedly exploit that $\zeta^{12}-\zeta^{6}+1=0$.  One has

\begin{align*}
\left|\frac{\zeta^{11}-\zeta^{7}+\zeta^5-\zeta^3+\zeta}{\zeta^{11}-2\zeta^{9}+\zeta^{7}-\zeta^{5}+\zeta^3} \right| &= \left | \zeta^8 - \zeta^2 \right | \\[0.07in]
&=\left| \zeta^{13}\cdot (\zeta^8 - \zeta^2) \right|\\[0.07in]
&=\left|-\zeta^{9}\right |\\[0.07in]
&=1
\end{align*}
Hence $m_1=m_2=1$.  Furthermore 
\begin{align*}
m_3=m_{4} &= \left| \frac{\zeta^{9}-\zeta^{3}+\zeta}{\zeta^{11}-\zeta^{7}+\zeta^{5}-\zeta^{3}+\zeta} \right | \\[0.07in]
&=\left | \zeta^{4}+ \zeta^{-4} \right| \\[0.07in]
&= 2 \cos \left  (\frac{2 \pi}{9} \right ) 
\end{align*}
\begin{align*}
 m_5=m_{6} &= \left| \frac{-\zeta^{11}+\zeta^{9}-\zeta^{7}}{\zeta^{11}-\zeta^{7}+\zeta^{5}-\zeta^{3}+\zeta} \right |  \\[0.07in]
&=\left | -\zeta^{2}-1 \right| \\[0.07in]
&=\left | -\zeta^{17}\cdot(\zeta^{2}+1) \right| \\[0.07in]
&=\left|\zeta^{1}+\zeta^{-1}\right |\\[0.07in]
&= 2\cos \left (\frac{\pi}{18}\right ) 
\end{align*}
\begin{align*}
m_{7}=m_{8} &= \left| \frac{\zeta^{7}+\zeta^{3}-\zeta}{\zeta^{11}-\zeta^{7}+\zeta^{5}-\zeta^{3}+\zeta} \right | \\[0.07in]
& =\left | -\zeta^{11} + \zeta^5 + \zeta^3 + \zeta \right| \\[0.07in]
& = \left | \zeta\cdot (\zeta^{2}+\zeta^{-2})\cdot (\zeta^{4}+\zeta^{-4}) \right| \\[0.07in]
& = 4\cos \left (\frac{\pi}{9}\right ) \cos \left (\frac{2 \pi}{9}\right)
\end{align*}
Recall that if $\sigma_{\textrm{Coxeter}}$ is a Coxeter element, the mass spectrum corresponds to the Perron-Frobenius eigenvector of the Cartan matrix. In particular, the mass spectrum has elegant trigonometric expressions.  We have shown that similar formulas hold for our choice of $\sigma$:

\def\arraystretch{1.2}
$$\begin{array}{c|c|c|c}
&\sigma_{\textrm{Coxeter}} & \sigma & \\
\hline   
\hline
\textrm{order} &$12$& $9$&\\
\hline
 m_1 &1  & 1&   \\
m_2 &1&  1&\\ 
m_3& 2 \cos \left (\frac{3\pi}{12} \right ) & 2\cos(\frac{2\pi}{9})&  \\
m_4&4\cos(\frac{\pi}{12})\cos(\frac{4\pi}{12})& 2\cos(\frac{2\pi}{9}) &\\
m_5&4\cos(\frac{\pi}{12})\cos(\frac{4\pi}{12}) &  2\cos(\frac{\pi}{18}) \\
m_6&4\cos\left (\frac{\pi}{12} \right ) \cos\left (\frac{3\pi}{12} \right ) &  2\cos(\frac{\pi}{18})  \\
m_7&&4\cos(\frac{\pi}{9})\cos(\frac{2\pi}{9})\\
m_8&&4\cos(\frac{\pi}{9})\cos(\frac{2\pi}{9})
\end{array}
$$
\def\arraystretch{1.0}
In Table \ref{example1graph} we plot the normalized masses for the Coxeter case,  denoted by $\textrm{E}_6$,  as well as for $\sigma$,  denoted by $\textrm{E}_{6}(\textrm{a}_{1})$ (since this is the conjugacy class of $\sigma$ in the notation of \cite{CAR}). 
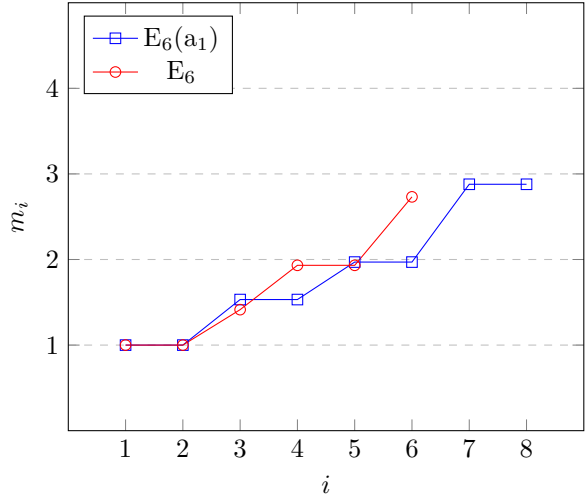
\begin{table}[hbt!]
\caption{}
\label{example1graph}
\centering
\begin{tikzpicture}[baseline=(current bounding box.center)]
\begin{axis}[
    title={Mass ratios},
    xlabel={$i$},
    ylabel={$m_i$},
    xmin=0, xmax=9,
    ymin=0, ymax=5,
    xtick={1,2,3,4,5,6,7,8},
    ytick={1,2,3,4},
    legend pos=north west,
    ymajorgrids=true,
    grid style=dashed,
]
\addplot[
	color=blue,
    mark=square,
    ]
    coordinates {
    (1,1)  (2,1)(3,1.532)(4,1.532)(5,1.970)(6,1.970)(7,2.879)(8,2.879)
    };
    \addlegendentry{$\textrm{E}_{6}(\textrm{a}_1)$}
\addplot[
	color=red,
    mark=o,
    ]
    coordinates {
    (1,1)  (2,1)(3,1.414)(4,1.932)(5,1.932)(6,2.732)
    };
   \addlegendentry{$\textrm{E}_{6}$}    
\end{axis}
\end{tikzpicture} \hspace{5mm}
\begin{tabular}{c|c}
$\textrm{E}_{6}$&$\textrm{E}_{6}(a_{1})$\\
\hline \\
1&1\\
1&1\\
1.414...&1.532...\\
1.932...&1.532...\\
1.932...&1.970... \\
2.732...&1.970...\\
&2.879...\\
&2.879...
\end{tabular}
\end{table}

\subsection{Example II}
\label{second-example}
Let $\mathfrak g = \mathfrak f_4$ and let $\alpha_1,\cdots,\alpha_4$ be simple roots, indexed as in \cite{BOURBAKI}. Consider the Weyl group element
$$\sigma=  r_{\alpha_1} r_{\alpha_3+\alpha_4}  r_{\alpha_2} r_{\alpha_1+\alpha_2+\alpha_3} $$ 
It is of order $6$ and by the same argument as in Example I one sees that $1$ is not an eigenvalue of $\sigma$.  For $\zeta_{6}=e^{2\pi i/6}$ the eigenvalues are
$$\zeta_{6}, \zeta_{6}, \zeta_{6}^{5}, \zeta_{6}^{5}$$
See \cite{BOU} (Table 1). The check for regularity is different than in Example I since both eigenvalues occur with multiplicity bigger than $1$.  One approach is to calculate explicitly an eigenvector that is not orthogonal to any root, along the lines of the calculations in the current section.  Instead,  we show in Section \ref{general-theory-section} that $\sigma$ lies in the conjugacy class $\textrm{F}_{4}(a_1)$, which is known to be regular.  Either way, the conditions of Theorem \ref{main-theorem} are satisfied.

Let $\zeta=\zeta_{24}=e^{2\pi i/24}$ be a primitive $24$'th root of unity, its minimal polynomial over $\mathbb{Q}$ is  $x^{8}-x^{4}+1$. We claim that
\begin{eqnarray}
\label{ExampleII-Lambda}
\Lambda_{+}= \zeta \alpha_{1} + (\zeta+\zeta^{-3}) \alpha_{2}+2\zeta \alpha_{3}+2\zeta \alpha_{4}
\end{eqnarray}
is an eigenvector of $\sigma$ with eigenvalue $\zeta_{6}$.  The $8$ orbits $\mathcal O_{1},\cdots,\mathcal O_{8}$ of the action of the cyclic group $\langle \sigma \rangle$ on the set of roots are as follows:  

\hspace{0.2in}

\def\arraystretch{1.3}

\begin{center}

\begin{tabular}{p{3cm}|p{3.5cm}|p{3.5cm}|p{3.5cm}}
$\mathcal O_{1}$ &$ \mathcal O_{2} $&$ \mathcal O_{3} $&$ \mathcal O_{4}$\\
\hline   
\hline 
&&& \\
 $\alpha_1 $&$  \alpha_{2} $&$ \alpha_{3}$&$ \alpha_{4}$\\

$-\alpha_2-2\alpha_3 $&$-\alpha_1-\alpha_2-2\alpha_3-2\alpha_4  $&$ \alpha_{1}+\alpha_{2}+2\alpha_3+\alpha_4 $&$ \alpha_2+\alpha_3+\alpha_4$ \\

$-\alpha_{1}-\alpha_{2}-2\alpha_{3}$&$ -\alpha_1-2\alpha_2-2\alpha_3-2\alpha_4 $&$ \alpha_{1}+\alpha_{2}+\alpha_3+\alpha_4 $&$ \alpha_2+\alpha_3$\\

$-\alpha_1$&$-\alpha_2 $&$-\alpha_3 $&$ -\alpha_4$
 \\

$\alpha_{2}+2\alpha_{3} $&$\alpha_1+\alpha_2+2\alpha_3+2\alpha_4 $  &$ -\alpha_{1}-\alpha_{2}-2\alpha_3-\alpha_4 $&$ -\alpha_2-\alpha_3-\alpha_4$\\

$\alpha_{1} +\alpha_2+2\alpha_3 $&$ \alpha_1+2\alpha_2+2\alpha_3+2\alpha_4   $&$ -\alpha_{1}-\alpha_{2}-\alpha_3-\alpha_4$ &$ - \alpha_2 - \alpha_3$\\

&&&
\end{tabular}

\end{center}

\begin{center}

\begin{tabular}{p{3cm}|p{3.5cm}|p{3.5cm}|p{3.5cm}}
$\mathcal O_{5} $&$ \mathcal O_{6}$ &$ \mathcal O_{7} $&$ \mathcal O_{8}$\\
\hline   
\hline 
&&& \\
$ \alpha_3+\alpha_4 $&$ \alpha_1+ \alpha_{2}+\alpha_3 $&$ \alpha_{1} + \alpha_2 $&$ \alpha_{1}+2\alpha_2+2\alpha_3$\\

$\alpha_1+2\alpha_2+3\alpha_3 +\alpha_4 $&$-\alpha_2-2\alpha_3-\alpha_4 $&$ -\alpha_{1}-2\alpha_{2}-4\alpha_3-2\alpha_4 $&$- \alpha_2-2\alpha_3-2\alpha_4$ \\

$\alpha_{1}+2\alpha_{2}+2\alpha_{3} +\alpha_4$&$ -\alpha_1-2\alpha_2-3\alpha_1-\alpha_4 $
 &$-2\alpha_1-3\alpha_2-4\alpha_3-2\alpha_4 $&$-\alpha_1-3 \alpha_2-4\alpha_3-2\alpha_4$\\

$-\alpha_3-\alpha_4$&$-\alpha_1-\alpha_2-\alpha_3 $&$-\alpha_1-\alpha_2 $&$ -\alpha_1-2\alpha_2-2\alpha_3$
 \\

$-\alpha_1-2\alpha_2-3\alpha_3 -\alpha_4$&$\alpha_2+2\alpha_3+\alpha_4$&$ \alpha_{1}+2\alpha_{2}+4\alpha_3+2\alpha_4 $&$ \alpha_2+2\alpha_3+2\alpha_4$\\

$-\alpha_{1}-2\alpha_{2}-2\alpha_{3} -\alpha_4$&$ \alpha_1+2\alpha_2+3\alpha_3+\alpha_4$   &$ 2\alpha_{1}+3\alpha_{2}+4\alpha_3+2\alpha_4 $&$ \alpha_1+3\alpha_2+4 \alpha_3+2  \alpha_4$\\

&&&
\end{tabular}

\end{center}

\def\arraystretch{1.0}

In particular
\begin{eqnarray*}
\sigma (\alpha_{1}) &=& -\alpha_2-2\alpha_{3}\\
\sigma (\alpha_{2}) &=& -\alpha_{1}-\alpha_{2}-2\alpha_{3} -2\alpha_{4}\\
\sigma (\alpha_{3}) &=&  \alpha_{1}+\alpha_{2}+2\alpha_{3}+\alpha_{4}\\
\sigma (\alpha_{4}) &=& \alpha_{2}+\alpha_{3}+\alpha_{4}
\end{eqnarray*}
It follows that
$$\sigma(\Lambda_{+}) = \zeta^{4}\cdot \Lambda_{+}=\zeta_{6} \cdot \Lambda_{+}$$
as desired.  Orbit representatives for the $\sigma$ action on the set of roots are given by
$$\alpha_1 \; ,  \; \alpha_2 \; , \; \alpha_3 \; , \; \alpha_4 \; , \;  \alpha_3+\alpha_4 \; ,\; \alpha_1+\alpha_2+\alpha_3 \; ,  \; \alpha_1+\alpha_2 \; ,  \; \alpha_1+2\alpha_2+2\alpha_3$$
Using Theorem \ref{main-theorem} one obtains the following $8$ masses, repeatedly using $\zeta^{8}-\zeta^{4}+1=0$:
\begin{align*}
|\Lambda_{+} \cdot \alpha_1 |&= |\zeta-\zeta^{-3}|=|\zeta^{5}|=1 \\[0.07in]
|\Lambda_{+} \cdot \alpha_2| &=|-\zeta+2\zeta^{-3}|=| \zeta^{5}\cdot (-\zeta+2\zeta^{-3})|=|\zeta^{2}+\zeta^{-2}|= 2 \cos\left( \frac{\pi}{6} \right)\\[0.07in]
|\Lambda_{+} \cdot \alpha_3| &=|-\zeta^{-3}|=1\\[0.07in]
|\Lambda_{+} \cdot \alpha_4| &=|\zeta |=1\\[0.07in]
|\Lambda_{+} \cdot (\alpha_3+\alpha_4)|&= |\zeta-\zeta^{-3}| =1 \\[0.07in]
|\Lambda_{+} \cdot (\alpha_1+\alpha_2+\alpha_3)|&= 0\\[0.07in]
|\Lambda_{+} \cdot (\alpha_1+\alpha_2)|&=|\zeta^{-3}|=1\\[0.07in]
|\Lambda_{+} \cdot (\alpha_1+ 2\alpha_2+2\alpha_3)|&=|-\zeta +\zeta^{-3}|=1
\end{align*}
As in Example I, we have shown that the Toda-Weyl mass spectrum for $\sigma$ has trigonometric expressions analogous to those for affine Toda theory:

\def\arraystretch{1.2} 
$$\begin{array}{c|c|c|c}
&\sigma_{\textrm{Coxeter}} & \sigma & \\
\hline   
\hline 
\textrm{order} &$12$& $6$&\\
\hline
m_1 &1   & 0&   \\
m_2 &2\cos\left (\frac{3 \pi}{12} \right )& 1&\\
m_3&2\cos(\frac{\pi}{12}) &  1&  \\
m_4&4\cos(\frac{\pi}{12})\cos(\frac{3 \pi}{12})& 1 &\\
m_5&  &  1\\
m_6& &  1 \\
m_7&&1\\
m_8 &&2   \cos\left (\frac{\pi}{6} \right ) 
\end{array}
$$
\def\arraystretch{1.0}
The masses are plotted in Table \ref{example2graph},  normalized so that the first non-zero mass equals $1$. 

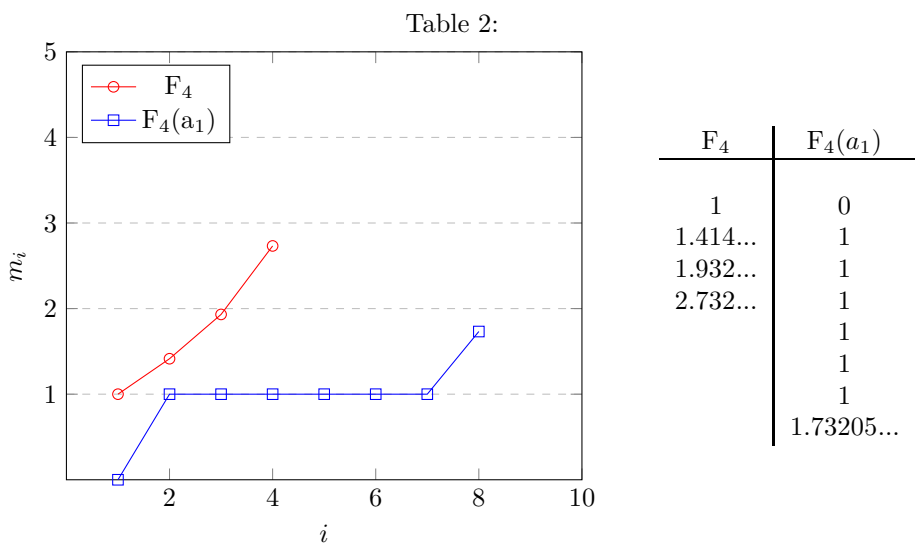
\begin{table}[hbt!]
\caption{}
\label{example2graph}
\centering
\begin{tikzpicture}[baseline=(current bounding box.center)]
\begin{axis}[
    title={},
    xlabel={$i$},
    ylabel={$m_i$},
    xmin=0, xmax=10,
    ymin=0, ymax=5,
    xtick={2,4,6,8,10},
    ytick={1,2,3,4,5},
     legend pos=north west,
    ymajorgrids=true,
    grid style=dashed,
]

\addplot[
	color=red,
    mark=o,
    ]
    coordinates {
     (1,1.00000000000000)(2, 1.41421356237310)(3, 1.93185165257814)(4, 2.73205080756887)  
    };
   \addlegendentry{$\textrm{F}_{4}$}  
      
\addplot[
	color=blue,
    mark=square,
    ]
    coordinates {
 (1,0.00000000000000)(2, 1.0000000000000)(3, 1.0000000000000)(4, 1.0000000000000)(5, 1.0000000000000)(6, 1)(7, 1)(8, 1.73205)
    };
    \addlegendentry{$\textrm{F}_{4}(\textrm{a}_1)$}

\end{axis}
\end{tikzpicture} \hspace{5mm}
\begin{tabular}{c|c}
$\textrm{F}_{4}$&$\textrm{F}_{4}(a_{1})$\\
\hline \\
1&0\\
1.414...&1\\
1.932...&1\\
2.732...&1\\
&1 \\
&1 \\
&1 \\
&1.73205...
\end{tabular}
\end{table}

\subsection{A general theory}
\label{general-theory-section}

As demonstrated in Example I and II,  Theorem  \ref{main-theorem} allows the effective calculation of the Toda-Weyl mass spectrum.  Nonetheless, the meaning of the spectrum might still be open. To this end, recall that in the Coxeter case a crucial insight is the relation to the Perron-Frobenius eigenvector of the Cartan matrix. With some mathematical effort, a corresponding theory can be developed for the Toda-Weyl Lagrangians. The full details will be presented elsewhere, but we describe the approach for the two previously considered examples.

In Example I,  the Weyl group element of $\mathfrak e_6$ is given by
$$\sigma= r_{\alpha_1}r_{\alpha_2}r_{\alpha_5} r_{\alpha_6} r_{\alpha_2+\alpha_4}r_{\alpha_3+\alpha_4} $$
For simple roots, the relative geometry is captured by the Cartan matrix and the Dynkin diagram.  Consider now analogous constructions that capture the geometry of the six roots 
$$\alpha_{1},\alpha_{2},\alpha_{5},  \alpha_{6}, \alpha_{2}+\alpha_{4}, \alpha_{3}+\alpha_{4}$$
involved in the definition of $\sigma$. Order them arbitrarily as $\gamma_1,
\cdots, \gamma_6$ and define the Carter matrix $K$ in complete analogy with 
the Cartan matrix via $K_{i,j} = 2 \cdot  (\gamma_i, \gamma_j)  / (\gamma_j 
, \gamma_j)$.  Now define a graph with vertices corresponding to the $
\gamma_i$'s and the $i$'th and $j$'th vertices are joined by $N_{i,j}$ lines 
where
$$N_{i,j}=K_{i,j}\cdot K_{j,i}$$
This graph generalizes the notion of Dynkin diagram and was introduced in seminal work by Carter in \cite{CAR}, classifying conjugacy classes of Weyl groups. In the present situation one obtains 
$$\xymatrix{
\alpha_2\ar@{-}[r] \ar@{-}[d] &\alpha_3 + \alpha_4  \ar@{-}[r] \ar@{-}[d]  & \alpha_1  \\
\alpha_2+\alpha_4 \ar@{-}[r] & \alpha_5   \ar@{-}[r] & \alpha_6
}$$
In the notation of \cite{CAR}, this means that $\sigma$ lies in the conjugacy class $\textrm{E}_{6}(a_1)$.  If we let $\gamma_1=\alpha_1, \gamma_2=\alpha_2, \gamma_3=\alpha_5, \gamma_4=\alpha_6,\gamma_5=\alpha_2+\alpha_4,  \gamma_6=\alpha_3+ \alpha_4$ 
one obtains the corresponding Carter matrix
$$K=\begin{bmatrix}
2&0&0&0&0&-1\\
0&2&0&0&1&-1\\
0&0&2&-1&-1&-1\\
0&0&-1&2 &0&0\\
0&1&-1&0&2&0\\
-1&-1&-1&0&0&2
\end{bmatrix}$$
Let $\zeta=e^{2\pi i/36}$.  A simple calculation yields an eigenvector $(x_1,\cdots,x_6)^{\textrm{T}}$ of $K$ with eigenvalue $\lambda = 2 -(\zeta^2+ \zeta^{-2})$:
$$x_1=1=x_{4}$$
$$x_2=\zeta^{10}-\zeta^8+1=1-2\cos \left (\frac{4 \pi}{9} \right ) =x_{5}$$
$$x_3=-\zeta^{10}+\zeta^{4}+\zeta^{2}=2\cos \left (\frac{\pi}{9} \right )=x_{6}$$
It turns out that the relative geometry of the roots $\gamma_{i}$ allows to write down the desired eigenvector $\Lambda_{+}$ of $\sigma$.  The definition of $\Lambda_{+}$ in Equation (\ref{exampleI-Lambda}) is simply 
$$\Lambda_{+} = \zeta \cdot (x_1 \gamma_1+ x_2 \gamma_2 + x_3 \gamma_3) + \zeta^{-1} \cdot (x_4 \gamma_4+ x_5 \gamma_5+ x_6 \gamma_6)$$
Therefore, the eigenvector of the Carter matrix determines the pairings of $\Lambda_{+}$ with a basis of root space, and hence via Theorem \ref{main-theorem} determines the mass spectrum. 

To make this even more clear, re-scale our choice of eigenvector $(x_1,\cdots,x_6)^{\textrm{T}}$ of $K$ such that the smallest entry is $1$, in other words we divide the original entries by $1-2\cos \left (\frac{4 \pi}{9} \right )$.  Ordered by size, the re-scaled entries are then
$$1,1,2\cos \left( \frac{2 \pi}{9} \right),2\cos \left( \frac{2 \pi}{9} \right) , 4\cos \left( \frac{ \pi}{9} \right) \cos \left( \frac{2\pi}{9}\right ),4\cos \left( \frac{ \pi}{9} \right) \cos \left( \frac{2\pi}{9}\right )$$
So $(\textrm{rank } \mathfrak g)= 6$ out of the $8$ masses of the Toda-Weyl theory associated to $\sigma$ come from the eigenvector of the Carter matrix.  In this manner, the celebrated relation between the affine Toda mass spectrum and the Perron-Frobenius eigenvector of the Cartan matrix is generalized to the element $\sigma$ in the conjugacy class $\textrm{E}_{6}(a_{1})$.

The same phenomenon persists in Example II where we consider the element of the Weyl group of $\mathfrak f_{4}$ given by
$$\sigma=  r_{\alpha_1} r_{\alpha_3+\alpha_4}  r_{\alpha_2} r_{\alpha_1+\alpha_2+\alpha_3} $$ 
The corresponding graph equals
$$\xymatrix @R=1.5pc @C=1.5pc {
\alpha_1\ar@<-0.3ex>@{-}[d] \ar@<0.3ex>@{-}[d]  \ar@{-}[r] & \alpha_2 \ar@<-0.3ex>@{-}[d] \ar@<0.3ex>@{-}[d]   \\
\alpha_1+\alpha_2+\alpha_3 \ar@{-}[r] & \alpha_3+\alpha_{4}
}$$
Let $\gamma_{1}=\alpha_1$, $\gamma_2=\alpha_3+\alpha_4$, $\gamma_3=\alpha_2$, $\gamma_4=\alpha_1+\alpha_2+\alpha_3$. The corresponding Carter matrix is then
$$K=
\begin{bmatrix}
2 &0 &-1 & 2\\
0&2&-1& -1\\
-1&-2&2& 0\\
1 & -1 &0 &2
\end{bmatrix} 
 $$
Note that it is non-symmetric and we now choose a left eigenvector! Let $\zeta=e^{2 \pi i/24}$ and let $\lambda=2-(\zeta^{2}+\zeta^{-2})$.  One possible corresponding left eigenvector of $K$ is given by $(x_{1},x_{2},x_{3},x_{4})^{\textrm{T}}$ with
\begin{align*}
x_{1}&=1\\[0.07in]
x_2&=2\\[0.07in]
x_3&=2\cos \left (\frac{\pi}{6} \right )\\[0.07in]
x_4&=0
\end{align*}
As in Example I, the desired eigenvector $\Lambda_{+}$ of $\sigma$ can be expressed in terms of this data: The choice of $\Lambda_{+}$ in Equation (\ref{ExampleII-Lambda}) is simply
$$\Lambda_{+} = \zeta \cdot (x_1 \gamma_1 + x_2\gamma_2) + \zeta^{-1}\cdot (x_3\gamma_3 + x_4 \gamma_4)$$
To realize part of the mass spectrum in terms of the eigenvector entries $x_1,\cdots,x_4$ requires more care than in example I: In non-simply laced cases an interesting duality occurs, the $\sigma$ eigenvector is expressed in term of a left eigenvector of $K$, whereas the inner products of $\Lambda_{+}$ with the roots $\gamma_i$ are expressed in terms of a right eigenvector. This duality is already present in the classical Coxeter case, see \cite{FLO} (Equation 24).  To make this explicit, let us normalize the root lengths so that $\alpha_1^2=\alpha_2^2=2$ and $\alpha_3^2=\alpha_4^2=1$. Then $\gamma_1^2 = \gamma_3^2 =2$ and $\gamma_2^2=\gamma_4^2=1$ and the right eigenvector corresponding to the left eigenvector $(x_{1},x_{2},x_{3},x_{4})^{\textrm{T}}$ has entries $\gamma_i^2 \cdot x_i$. Its entries are therefore
$$2,2, 4\cos \left( \frac{\pi}{6} \right ),0$$
Hence, after scaling to make the lowest entry $1$, one sees that $(\textrm{rank } \mathfrak g)=4$ of the Toda-Weyl masses calculated in Section \ref{second-example} can be read off from a suitable eigenvector of the Carter matrix.

This relation between the relative geometry of special sets of roots (the ``Carter roots'' $\gamma_{i}$) and eigenvectors of Weyl group elements allows to generalize the mass description of affine Toda theories in terms of the Perron-Frobenius eigenvector of the Cartan matrix. We will describe the full mathematical details elsewhere.

\section{Conclusions}

Starting with the formulation of affine Toda theory in Equation (\ref{crucial-Lagrangian}), we generalized this set-up by considering Lagrangians involving eigenvectors $\Lambda_{+}$ of other Weyl group elements $\sigma$.

Note that the resulting Toda-Weyl theory does not usually have a simple description of the form as in Equation (\ref{new-Lagrangian}): Typically, when the Weyl group element eigenvector is described in terms of root spaces,  it involves two roots $\xi$ and $\nu$ such that $\xi-\nu$ is again a root. This is one of the reasons we simply took the formulation of affine Toda theory in Equation (\ref{crucial-Lagrangian}) as our starting point.

Under some technical conditions on $\sigma$ that are frequently satisfied, we obtained a description of the classical mass spectrum in terms of the pairings of $\Lambda_{+}$ with orbit representatives of the action of $\langle \sigma \rangle$ on the set of roots. After calculating the masses in some illustrative examples,  we sketched in Section \ref{general-theory-section} how one can construct the desired eigenvectors $\Lambda_{+}$ in terms of eigenvectors of generalized Cartan matrices. We will describe the full mathematical details elsewhere.  This relation between the linear algebra of Weyl group elements and matrices describing the relative geometry of special sets of roots generalizes the celebrated description of the affine Toda mass spectrum in terms of a Perron-Frobenius eigenvector of the relevant Cartan matrix.

There are many open questions regarding these Toda-Weyl theories. For example, one should calculate the three point couplings, as is done in usual Toda theory in \cite{FLO}. Another open question concerns the integrability or failure thereof. We do not address this issue here at all but hope to return to it on a future occasion.

$$$$

\noindent
\textbf{Acknowledgments:}

It is a pleasure to thank the referee for comments that improved the exposition.

\hspace{0.2in}

\noindent
\textbf{Declaration of competing interest:}

The author declares to have no known competing financial interests or personal relationships that could have appeared to influence the work reported in this paper.

\hspace{0.2in}

\noindent
\textbf{Data availability:}

No data was used for the research described in the article.

\end{document}